\documentclass[11pt]{article}

\usepackage{endnotes}
\let\footnote=\endnote

\usepackage{natbib}
\bibpunct[, ]{(}{)}{,}{a}{}{,}%
\usepackage{url}
\usepackage{nicefrac}
\usepackage{amssymb}
\usepackage{amsmath,amsopn}
\usepackage{amsthm}
\usepackage{amsfonts}
\usepackage{amssymb}
\usepackage{hyperref}
\hypersetup{
	colorlinks,
	linkcolor={red!50!black},
	citecolor={blue!50!black},
	urlcolor={blue!80!black}
}
\usepackage{breakurl}
\newcommand{\LeftEqNo}{\let\veqno\@@leqno}
\usepackage{environ}
\usepackage{xcolor}
\usepackage{float}
\usepackage{graphicx}
\usepackage[noend]{algpseudocode}
\usepackage{algorithm}
\usepackage{tikz}
\usepackage{float}
\usepackage{array}
\usepackage{enumerate}
\usepackage{subcaption}
\usepackage{booktabs}
\usepackage{comment}
\usepackage{stackrel}
\usepackage{color,soul}
\usepackage{algorithmicx}
\usepackage{graphicx}
\usepackage{xspace}

\usepackage{footnote}
\usepackage{geometry}

\usepackage{natbib}
\bibpunct[, ]{(}{)}{,}{a}{}{,}%

\usepackage{authblk}
\usepackage{lipsum}

\makeatletter
\renewcommand\footnoterule{%
	\kern-3\p@
	\hrule\@width \textwidth
	\kern2.6\p@}
\makeatother

\makeatletter
\renewcommand*{\@fnsymbol}[1]{\ensuremath{\ifcase#1\or *\or \dagger\or \ddagger\or **\or \mathsection\or \mathparagraph\or \|\or  \dagger\dagger
		\or \ddagger\ddagger \else\@ctrerr\fi}}
\makeatother

\usepackage{hyperref}	 
\hypersetup{
	colorlinks={true},linkcolor={magenta},citecolor=blue!60!black,urlcolor  = blue!60!black,
}

\newtheorem{proposition}{Proposition}
\newtheorem{definition}{Definition}
\newtheorem{corollary}{Corollary}

\newtheorem{lemma}{Lemma}

\newcommand{\halmos}{\ensuremath{\square}}
\newcommand{\Halmos}{\ensuremath{\square}}

\newcommand{\CBPP}{\textbf{CBPP}}

\newcommand{\SSP}{\textbf{SSP}}
\newcommand{\SSPD}{\textbf{SSPD}}

\newcommand{\BPMCF}{\textbf{BPMCF}}

\newcommand{\BPP}{\textbf{BPP}}

\newcommand{\BPPC}{\textbf{BPPC}}

\newcommand{\firstIP}{\textbf{(IP)}}

\newcommand{\objects}{\ensuremath{\mathcal{O}}}
\newcommand{\object}{\ensuremath{o}}

\newcommand{\objcost}{\ensuremath{v}}

\newcommand{\objcostd}{\ensuremath{v^d}}

\newcommand{\size}{\ensuremath{s}}

\newcommand{\colour}{\ensuremath{g}}
\newcommand{\colours}{\ensuremath{G}}

\newcommand{\numbins}{\ensuremath{k}}

\newcommand{\binsize}{\ensuremath{B}}

\newcommand{\bins}{\ensuremath{ \mathcal{B} }}
\newcommand{\bin}{\ensuremath{ b }}

\newcommand{\sizebin}{\ensuremath{ B }}

\newcommand{\itemsizes}{\mathcal{K}}
\newcommand{\itemsize}{k}

\newcommand{\groups}{\mathcal{G}}
\newcommand{\group}{g}

\algnewcommand{\LineComment}[1]{\State \(\triangleright\) #1}

\begin{document}

\title{Binary Decision Diagrams for Bin Packing with \\ Minimum Color Fragmentation}

\author[1]{David Bergman\thanks{david.bergman@uconn.edu }}
\author[2]{Carlos Cardonha\thanks{carloscardonha@br.ibm.com}}
\author[1]{Saharnaz Mehrani\thanks{saharnaz.mehrani@uconn.edu}} 
\affil[1]{\small Department of Operations and Information Management, University of Connecticut}
\affil[2]{\small IBM Research}
\maketitle

\begin{abstract}	
		Bin Packing with Minimum Color Fragmentation (\BPMCF) is an extension of the Bin Packing Problem in which each item has a size and a color and the goal is to minimize the sum of the number of bins containing items of each color.  
		In this work, we introduce the~\BPMCF{} and present a decomposition strategy to solve the problem, where the assignment of items to bins is formulated as a binary decision diagram and an optimal integrated solutions is identified through a mixed-integer linear programming model. Our computational experiments show that the proposed approach greatly outperforms a direct formulation of~\BPMCF{} and that its performance is suitable for large instances of the problem.
\end{abstract}	
	
\section{ Introduction}

In this work, we investigate Bin Packing with Minimum Color Fragmentation (\BPMCF{}), an extension of the Bin Packing Problem in which each item is associated with a \textit{color} and one wishes to identify assignments of items to bins placing items of a common color in the fewest number of bins possible.
\BPMCF{}  finds application in scenarios where having elements of a certain category (color) in each bin results in costs that should be minimized. In production planning, for example, a company requires expensive modules in order to process orders each of which is of a certain type and has a processing time.  Each plant has limited total processing time.  In this situation, items correspond to orders and order types correspond to categories.  To minimize costs for the modules, the company strives to pack orders of each type in as few  bins as possible.  Similarly, in logistics, there are scenarios where heterogeneous cargo needs to be transported and differences in cargo necessitate vehicles equipped with specific capabilities, for example temperature or pressure controlling devices.  Additionally, in large-group seating problems that arise in corporate events or weddings, attendees form subgroups and event organizers strive to pack individuals of the same type in the fewest number of tables.

Many colored extensions of the bin packing problem have been studied in the literature, but  with a different (typically opposite) objective. Examples abound. 
In the \emph{colored bin packing problem} (\CBPP{}), 
commonly colored items are not allowed to be packed next to each other in the same bin (\cite{balogh2015online,balogh2012black,bohm2014online}).
\CBPP{} has also been investigated  under the name of \emph{class constrained bin packing problem}, mostly from a theoretical perspective (\cite{shachnai2001polynomial,shachnai2004tight,xavier2008class}). Approximation results have  been obtained for the \emph{variable class constrained bin packing problem}, where bins may have different sizes and the goal is to minimize the sum of the sizes of the bins that have been used (\cite{dawande2001variable,xavier2008class}).
Jansen introduced the \emph{bin packing problem with conflicts} (\BPPC{}), a generalization of~\CBPP{}  where we are given a graph on the items, with edges indicating pairs of elements that  cannot be placed in the same bin~(\cite{jansen1999approximation,jansen1997approximation}).
Several algorithms have been introduced in the literature to address the~\BPPC{} (\cite{elhedhli2011branch,gendreau2004heuristics,muritiba2010algorithms,sadykov2013bin}), such as a branch-and-price algorithm  for general conflict graphs (\cite{sadykov2013bin}).
Another variant of~\CBPP{} is the \emph{co-printing problem}, in which bins are bounded both in terms of weight and number of colors they may contain (\cite{peeters2004co}); both heuristic and exact algorithms have been proposed to solve this problem (\cite{kochetov2017vns,kondakov2018core}). 
To the best of our knowledge, \BPMCF{} is yet to be investigated in the literature yet.


In this article, we introduce~\BPMCF{}, show how to cast the assignment of items to bins as a binary decision diagram (BDD) (\cite{Bry86,Bry92}), and present a mixed-integer linear programming (MIP) formulation to solve the problem. BDDs and their multivalued extension have been successfully applied in different applications for optimization (\cite{andersen2007constraint,miller1998implementing,BerHoeHoo11,matsumoto2018decision}). Decomposition strategies relying on the combination of decision diagrams and integer programming, such as the one employed in this work, have been applied to other optimization problems (\cite{bergman2016decomposition,bergman2018decision,BerCir17,LozBerSmi18}). Our experiments suggest the efficiency of the proposed algorithm, with a clear superiority over a direct  MIP formulation of~\BPMCF{}.



\section{Problem overview}\label{sec:Description}

A formal definition of~\BPMCF{} is presented below:


\begin{definition}[Bin Packing with Minimum Color Fragmentation]\label{def:BPMCF}
	Let~$\numbins \in \mathbb{N}$ denote the number of bins, $\binsize \in \mathbb{N}$ be the capacity of each bin, $\colours \subseteq \mathbb{N}$ be a set of colors, and $\objects  = \{\object_1,\object_2,\ldots,\object_n\}$ be a set  of indivisible items such that, for each $\object \in \objects$, $\size(\object) \in \mathbb{N}$ denotes its size and
	$\colour(\object) \in \colours$ its color. A feasible solution for the problem consists of a partition of~$\objects$ into disjoint sets~$\objects_1,\objects_2,\ldots,\objects_\numbins$ 
	such that $\sum_{\object \in \objects_i}\size(\object) \leq \sizebin$, $\forall i \in [k]$. 
	Let~$n_\colour$ denote the number of bins containing items of color~$\colour$. A feasible solution is said to be optimal if it minimizes~$\sum_{\colour \in \colours}n_\colour$. 
\end{definition}


Set~$\objects_\colour = \{ \object \in \objects: \colour(\object) = \colour \}$ contains the items in~$\objects$ whose color is~$\colour$, $\objects_\size = \{ \object \in \objects: \size(\object) = \size\}$ is the set of items of size~$\size$, and set $\objects_{\colour,\size} = \{ \object \in \objects_\colour: \size(\object) = \size\} = \objects_\colour \cap \objects_\size$ contains all items of color~$\colour$ and size~$\size$.  We use~$\bins$ to denote the set of bins available for packing.  $\itemsizes_\group = \{ s \in \mathbb{N}: \exists \object \in \objects_\group \text{ with } \size(\object) = s\}$ is the set of all sizes of items belonging to~$\objects_\group$. 
In scenarios where bins may have different sizes, we use~$\binsize_\bin$ to represent the capacity of bin~$\bin$.

\BPP{} is NP-hard (\cite{garey1978strong,garey2002computers}), and the following results shows that~\BPMCF{} is also challenging in the general but easy in a special case.
\begin{proposition}
	Deciding whether an arbitrary instance of~\BPMCF{} with at least 2 bins is feasible is NP-complete.	 
\end{proposition}
\begin{proof}
	The result follows from a reduction of the Partition Problem (\textbf{PP}), which is NP-complete (\cite{karp1972reducibility}). Given an instance~$I$ of \textbf{PP} with a set~$A$ of elements, with each~$a \in A$ of size~$s'(a)$, we create an instance~$I'$ of~\BPMCF{} such that each item~$i$ in~$I'$ is associated with an element~$i(a)$ of~$A$, has size~$\size(i(a))) = s'(a)$, and group~$\colour(i(a)) = 1$, i.e., all items have the same group. Moreover, $I'$ consists of two bins, each of size $\frac{1}{2}\sum_{a \in A}s'(a)$. A feasible solution for~$I$ can be directly converted into a solution for~$I'$ and vice-versa, so the reduction follows.	
	$\halmos$
\end{proof}

\begin{proposition}\label{prop:easy2bins}
	If each bin can contain at most two items, an optimal solution  of~\BPMCF{} can be computed in polynomial time.	 
\end{proposition}
\begin{proof}
	We can assume without loss of generality that each item~$\object$ in~$\objects$ is such that there is at least one other item~$\object'$ for which  $\size(\object) +\size(\object') \leq \binsize$.
	A feasible solution should contain at least $q = |\objects|-|\bins|$ pairs of items being placed in the same bin. For each possible value of~$q$, we create the following instance of the maximum weighted matching problem. Let $G = (V,E,w)$ be a graph where $V = V(\objects) \cup V'$,  each vertex in~$V(\objects)$ is associated with an item in~$\objects$ and vertices in~$V'$ contains $|\objects| - 2q$  artificial elements.
	Set~$E$ contains an edge for each pair $\{u,v\} \subseteq V(\objects)$; if $\colour(u) = \colour(v)$, $w(\{u,v\}) = 2+|\objects|$, and otherwise we have $w(\{u,v\}) = 1+|\objects|$.
	$E$ also contains an edge for each pair $\{u,v\} \in V(\objects) \times V'$, each with weight~$|\objects|^2$. 
	
	By construction, any optimal solution contains a set of edges covering all the artificial vertices, and solutions with~$q$ edges covering all the remaining elements have higher value than others with $q-1$ or less. The maximum weighted matching problem can be computed in polynomial time~(\cite{edmonds1965paths}), and the number of values of~$q$ that need to be inspected is bounded by~$|\objects|$, so the result follows.
	$\halmos$
\end{proof}

\section{Direct Formulation}\label{sec:MIPModels}

The following binary program is the direct formulation of~\BPMCF{} used for a baseline algorithm in our computational experiments. 
\[
\begin{array}{lllll}\label{Model:IP1}
\firstIP{} \quad & \mathrm{min}\;\;  \sum\limits_{ (\bin,\colour) \in \bins \times \colours } y_{\bin, \colour} & && \\

&\sum\limits_{\bin \in \bins} x_{\bin,\object}  &=& 1 &\forall \object \in \objects \\

&\sum\limits_{\object \in \objects} \size(\object) x_{\bin,\object}  &\leq& \binsize
&\forall \bin \in \bins \\

&x_{\bin,\object}  &\leq& y_{\bin,\colour} 
&\forall (\bin,\colour,\object) \in \bins \times \colours \times \objects_\colour \\

&x_{\bin,\object} \in \{0,1\} &&& \forall (\bin,\object) \in \bins \times \objects \\
&y_{\bin,\colour} \in \{0,1\} &&& \forall (\bin,\colour) \in \bins \times \colours
\end{array}
\]
In~\firstIP{}, the assignment of each item $ \object$ to each bin $\bin$ is defined by binary decision variable $  x_{\bin,\object} $. Additionally, we use~$ y_{\bin, \colour} $ to indicate whether bin $ \bin $ contains at least one item of color $ \colour $. The first family of constraints of~\firstIP{} asserts that each item is assigned to exactly one bin. The second family of constraints avoids assignments where the sum of the sizes of the selected items exceeds the capacity of the bin. 
The last set of constraints is used to set~$  y_{\bin,\colour} $; 
if $ x_{\bin,\object} = 1$ for some $\object \in \objects_{\colour}$, $  y_{\bin,\colour}  = 1$, whereas the objective function drives $  y_{\bin,\colour} $ to zero otherwise.


\newcommand{\BDD}{D}

\newcommand{\arcs}{A}
\newcommand{\arc}{a}

\newcommand{\BDDnodes}{N}
\newcommand{\BDDnode}{u}

\newcommand{\layers}{L}
\newcommand{\layer}{l}

\newcommand{\BDDpath}{p}

\newcommand{\pred}{\pi}

\newcommand{\arclabel}{q^{\binsize}}
\newcommand{\arcweight}{w}

\section{Binary decision diagram-based algorithm}\label{sec:BDDModel}

Our algorithm relies on a decomposition strategy in which the assignment of items to a bin is represented as a BDD and feasible solutions are given by the paths connecting the root node to the terminal node. The structure of a BDD depends solely on the capacity~$\binsize$ of the associated bin; we use~$\BDD^{\binsize}$ to refer to the BDD for bins of size of~$\binsize$ 
(dropping $B$ if it is clear from context).

\paragraph{Notation:} A BDD $\BDD^{\binsize}=(\BDDnodes^{\binsize},\arcs^{\binsize},v^{\binsize}, d^{\binsize})$ is a layered-acyclic graph composed of a set of nodes~$\BDDnodes^{\binsize}$, a set of arcs~$\arcs^{\binsize}$, together with a cost function
$v^{\binsize}:\arcs^{\binsize} \rightarrow \mathbb{R}$ and arc-domain function $d^{\binsize}:\arcs^{\binsize} \rightarrow \{0,1\}$ defined on the arcs.  Nodes in $\BDDnodes^{\binsize}$ are partitioned into a set $\layers^{\binsize} =\{0,1,...,|\objects|,|\objects|+1\}$ of layers. For every node $\BDDnode \in \BDDnodes^\binsize$, $\layer^\binsize(\BDDnode)$ denotes the layer where~$\BDDnode$ belongs. Layer 0 contains only the root node $r^\binsize$ of $\BDDnodes^{\binsize}$, and layer = $|\objects|+1$ contains its terminal node~$t^\binsize$. Each layer~$\layer \in \{1,...,|\objects|\}$ is associated with an item~$\object(\layer)$; analogously, we define $\object(\BDDnode) = \object(\layer^\binsize(\BDDnode))$ for each node~$\BDDnode \in \BDDnodes$. We assume the layers are ordered by groups first (any arbitrary ordering of~$\groups$ may be employed) and then arbitrarily in each group.
Each arc $\arc \in \arcs^{\binsize}$ is directed from a start-node~$\BDDnode^s(\arc)$ located in layer $\layer^\binsize(h(\arc)) \in \layers^\binsize \backslash t$ to an end-node $\BDDnode^e(\arc)$ located in layer $\layer^\binsize(\BDDnode^e(\arc))=\layer^\binsize(\BDDnode^s(\arc))+1$. That is, each arc connects two nodes in consecutive layers. Every node $\BDDnode \in \BDDnodes^{\binsize} \backslash \{0^\binsize\}$ has at most two arcs directed out of it with unique arc domains, referred to as the nodes \emph{one-arc} if $d(a) = 1$ and its \emph{zero-arc} if $d(a) = 0$.

In our BDD formulation, each arc~$\arc$ represents the decision about the inclusion of the item associated with its start-node~$\BDDnode^s(\arc)$ determined by the arc-domain of $\arc$.  Every root-to-terminal arc-specified path therefore encodes a collection of items defined by the one-arcs on the path.  Formally, let $p = (a^1, \ldots, a^{|\objects|})$ be a $r^B$ to $t^B$ path.  This path corresponds to the collection of items $\objects(p) := \{ o(l) : d(a^l) = 1 \}$ and has \emph{cost} $C(p) := \sum_{l=1}^{|\objects|} v^{\binsize}(\arc^l)$. Each BDD $D^B$ will satisfy the following two properties. (1) There is a one-to-one mapping between collection of items $\tilde{\objects} \subseteq \objects$ with $\sum_{o \in \tilde{\objects} } s(o) \leq B$ and root-to-terminal paths $p$ with $\objects(p) = \tilde{\objects}$.  (2) For every root-to-terminal path $p$, $C(p)$ equals the number of groups present in $\objects(p)$ (i.e., $C(p) = \left| \{ g : g(o) = g \textrm{ for } o \in \objects(p) \} \right|$).  Any such BDD is called an \emph{exact} BDD for a bin of size $B$. 

Given an instance of $\BPMCF$, suppose we build such a BDD for each bin; i.e., for $b \in \mathcal{B}$, we build $D^{B_b}$, which is an exact BDD for a bin of size $B_b$.  Any collection of paths $p^i$, one from each BDD, for which the family of sets $\objects(p^b), b \in \bins$ are mutually exclusive and collectively exhaustive of $\objects$ will provide a feasible solution where items in $\objects(p^i)$ are placed in bin $b$. Additionally, $\sum_{b \in \bins} C(p^\bin)$ will correspond to the objective function value of the solution.     The $\BPMCF$ can therefore be solved by finding such a collection of BDDs, and identifying a mutually exclusive and collectively exhaustive collection of paths that minimizes $\sum_{b \in \bins} C(p^\bin)$, a problem recently introduced into the literature as the \emph{consistent path problem} (\cite{BerCir17,LozBerSmi18,bergman2018decision,1807.09676}).

\paragraph{Construction:} Building BDDs for discrete optimization problems has been a major research focus in the last decade (\cite{HadHooOsulTieStu08,BerHoeHoo11,DDBook,Bergman:2016:DOD:3214682.3214686,10.1007/978-3-642-29828-8_3}).  Given a bin of size $B$, the family of sets of items that fit in the bin corresponds exactly to the set of feasible solutions to a knapsack problem, where each item $o$ has size $s(o)$ and the knapsack has capacity $B$.  The dynamic programming-based construction algorithm of a BDD representing such a set of solutions is well known (\cite{Trick2003,Behle2008,10.1007/978-3-319-44953-1_6,1802.08637}).  We adopt this algorithm, but with additional care required because of the objective function. 


We build a BDD iteratively, layer-by-layer, by associating with each node $u$ a \emph{state} $z(u) \in \mathbb{Z} \cup \mathbb{B}$ that will represent the remaining capacity in the bin for any partial solution defined by paths starting from the root and ending at $u$ and whether or not any item with the color of the object in layer $l(u)$ is selected.  The root node $r$ is assigned state $(B,0)$.  The first index represents the remaining capacity of the bin, and the second represents whether or not an item of the same color as the item $l(o)$ has been selected or not.  Recall that the items corresponding to the layers of the BDD are ordered by color. 

Having constructed layer $l$, we build layer $l+1$ by iteratively processing the nodes in $l$. When processing node $u$ with state $z(u) = (A,d)$, we calculate the states $z_0$/$z_1$ that will arise from a zero-arc/one-arc directed out of $u$.   $z_0 = (A,d)$ and $z_1 = (A-s(l(o)),1)$. Since a zero-arc dictates the omission of item $o$, neither $A$ of $d$ changes.  A one-arc dictates the inclusion of item $o$ and so the remaining capacity is $A - s( l(o) )$ and an object with this color is selected. 

For any state $z_0$ we create a zero-arc $a_0$ directed to a node $u_0$, and assign that node state $s(u_0) = z_0$.  If $A-s( l(o) ) < 0$, we do not create a one-arc, because this would exceed the capacity of the bin.   If $A-s( l(o) ) \geq 0$, we create a one-arc $a_1$ directed to a node $u_1$, and assign that node state $s(u_1) = z_1$.  The cost of the arc is 0 for $a_0$ and the cost of arc $a_1$ is 1 if the second component of $s(u)$ is 0, and is 0 otherwise.  This particular assignment of costs is because only when a one-arc results in the first item of a color being selected should a cost be assigned, as this will indicate that there are items from this color category that are selected for this bin should this arc be selected.    

If the item corresponding to layer $l$ is the last item with color $g$, we now change the second component of both $s(u_0)$ and $s(u_1)$ to 0, because in the subsequent layer $l+1$ no objects of color $g(o(l+1))$ have previously been considered. 

Finally, for $d=0,1$, if any node $\tilde{u}$ previously added to layer $l+1$ has state $z_d$, the $d$-arc directed out of $u$ is directed to $\tilde{u}$.  Otherwise, a new node $\hat{u}$ is created, added to layer $l+1$, with the $d$-arc directed from $u$ to $\hat{u}$.  After constructing all layers, each node in layer $l = n+1$ is merged into a single terminal node.  
\paragraph{Example:}  Consider an instance with 5 items of sizes 2,3,2,3,2 and colors 1,1,1,2,2, respectively, with bin capacity 4.   A BDD for this bin is presented in Figure~\ref{f:example}.  Each layer corresponds to an item.  The solid / dashed arcs correspond to one-arcs / zero-arcs. The arc costs are specified next to each arc.  Note that there can be one arcs with zero cost (e.g., the arc from $u_5$ to $u_9$).  Also, any solution corresponds to a path.  For example, selecting items $o_1$ and $o_3$ is a feasible solution that corresponds to the arc-directed path 
$r-u_2-u_5-u_9-u_{14}-t$.

\begin{figure}[h!]
	\centering
	\begin{tikzpicture}[scale=0.4][font=\sffamily, \tiny]

	\node (x1) at (14,-1) {$o_1$};
	
	\node (x2) at (14,-3) {$o_2$};
	
	\node (x3) at (14,-5) {$o_3$};
	
	\node (x4) at (14,-7) {$o_4$};
	
	\node (x5) at (14,-11) {$o_5$};
	
	\node (1) [draw,circle] [label=0:{\tiny(4,0)},draw] at (0,0) {r};

	\node (2) [draw,circle] [label=180:{\tiny(4,0)},draw] at (-2,-2) {$u_1$};
	
	\node (3) [draw,circle] [label=0:{\tiny(2,1)},draw] at (2,-2) {$u_2$};
	
	\node (4) [draw,circle] [label=180:{\tiny(4,0)},draw] at (-4,-4) {$u_3$};

	\node (5) [draw,circle] [label=0:{\tiny(1,1)},draw] at (0,-4) {$u_4$};
	
	\node (6) [draw,circle] [label=0:{\tiny(2,1)},draw] at (4,-4) {$u_5$};
	
	\node (7) [draw,circle] [label=180:{\tiny(4,0)},draw] at (-6,-6) {$u_6$};
	
	\node (8) [draw,circle] [label=180:{\tiny(2,0)},draw] at (-2,-6) {$u_7$};
	
	\node (9) [draw,circle] [label=180:{\tiny(1,0)},draw] at (2,-6) {$u_8$};
	
	\node (10) [draw,circle] [label=180:{\tiny(0,0)},draw] at (6,-6) {$u_9$};
	
	\node (11) [draw,circle] [label=180:{\tiny(4,0)},draw] at (-8,-8) {$u_{10}$};
	
	\node (12) [draw,circle] [label=180:{\tiny(1,1)},draw] at (-4,-8) {$u_{11}$};
	
	\node (13) [draw,circle] [label=180:{\tiny(2,0)},draw] at (0,-8) {$u_{12}$};
	
	\node (14) [draw,circle] [label=180:{\tiny(1,0)},draw] at (4,-8) {$u_{13}$};
	
	\node (15) [draw,circle] [label=180:{\tiny(0,0)},draw] at (8,-8) {$u_{14}$};
	
	\node (16) [draw,circle] [label=270:{\tiny(0,0)},draw] at (0,-12) {t};
	
	\path[-](1) edge [dashed] node [left] {0} (2);
	
	\path[-](1) edge node [right] {1} (3);
	
	\path[-](2) edge [dashed] node [left] {0} (4);
	
	\path[-](2) edge node [right] {1} (5);
	
	\path[-](3) edge [dashed] node [right] {0} (6);
	
	\path[-](4) edge [dashed] node [left] {0} (7);
	
	\path[-](4) edge node [right] {1} (8);
	
	\path[-](5) edge [pos = 0.15, dashed] node [right] {0} (9);
	
	\path[-](6) edge [pos=0.25, dashed] node [above] {0} (8);
	
	\path[-](6) edge node [right] {0} (10);
	
	\path[-](7) edge [dashed] node [left] {0} (11);
	
	\path[-](7) edge node [right] {1} (12);
	
	\path[-](8) edge [dashed] node [right] {0} (13);
	
	\path[-](9) edge [dashed] node [right] {0} (14);
	
	\path[-](10) edge [dashed] node [right] {0} (15);
	
	\path[-](11) edge [dashed, bend right] node [below] {0} (16);
	
	\path[-](11) edge node [below] {1} (16);
	
	\path[-](12) edge [dashed] node [below] {0} (16);
	
	\path[-](13) edge [dashed,bend right] node [left] {0} (16);
	
	\path[-](13) edge [solid,bend left] node [right] {1} (16);
	
	\path[-](14) edge [dashed] node [right] {0} (16);
	
	\path[-](15) edge [dashed] node [right] {0} (16);
	
	\end{tikzpicture}			
	\caption{Example BDD}
	\label{f:example}
\end{figure}
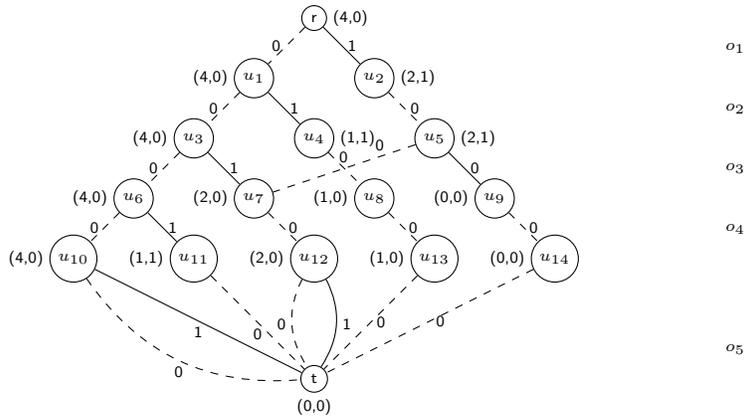

\paragraph{Network flow model:}

The BDDs allow us to formulated the consistent path problem through~(\textbf{ANF}), an \emph{Arc-based Network Flow} MIP to solve the problem. 
\[
\begin{array}{lllll}\label{Model:NetFlow}
(\textbf{ANF}) \quad & \mathrm{min}\;\; \sum\limits_{\bin \in \bins}  \sum\limits_{ \arc \in \arcs^{\binsize_\bin} } v(\arc)  y_{\bin,\arc} & && \\
& 
\sum\limits_{\substack{ \arc \in \arcs^{\binsize_\bin}; \\ \BDDnode^e(\arc) = \BDDnode }} y_{\bin,\arc} -  \sum\limits_{\substack{ \arc \in \arcs^{\binsize_\bin}; \\ \BDDnode^s(\arc) = \BDDnode }} y_{\bin,\arc} &=& 0 &\forall \bin \in \bins, \BDDnode \in \BDDnodes \backslash \{0^{\binsize_\bin},t^{\binsize(\bin)}\} \\
& 
\sum\limits_{\substack{ \arc \in \arcs^{\binsize(\bin)}; \\ \BDDnode^s(\arc) = 0^{\binsize_\bin} }} y_{\bin,\arc}   &=&  1 &\forall \bin \in \bins\\
& 
\sum\limits_{\substack{ \arc \in \arcs^{\binsize_\bin}; \\ \BDDnode^e(\arc) = t^{\binsize_\bin} }} y_{\bin,\arc}   &=&  1 &\forall \bin \in \bins\\
& 
\sum\limits_{\bin \in \bins} \sum\limits_{\substack{ \arc \in \arcs^{\binsize_\bin}; \\ \colour(\layer_{\BDDnode^e(\arc)}^{\binsize_\bin}) = \colour, \itemsize(\layer_{\BDDnode^e(\arc)}^{\binsize(\bin)}) = \itemsize }} \arcweight(\arc) y_{\bin,\arc}   &=&  |\objects_{\colour,\itemsize}| &\forall (\colour,\itemsize) \in \colours \times \itemsizes_\colour\\
& y_{\bin,\arc} \in \{0,1\} & && \forall \bin \in \bins,\arc \in \arcs^{\binsize_\bin }
\end{array}
\]
(\textbf{ANF}) employs binary variables~$y_{\bin,\arc}$, which indicate whether arc~$\arc$ composes the path selected for~$\BDD^\bin$. The first three families of equalities model the network flow constraints for each bin~$\bin$ in~$\bins$. The last family of constraints asserts that each item is picked exactly once, so they are the joint constraints of~(\textbf{ANF}).


\section{Computational Experiments}\label{sec:Experiments}
All code used to evaluate the model and algorithms presented were implemented in C++ and utilize commercial software Gurobi 8.0.0 to implement \textbf{IP} and \textbf{ANF}~(\cite{gurobi}); we used all default settings of the solver. All experiments were executed on an Intel CPU Core i7-4770 with 3.4 GHz, 32 GB of RAM. Each execution was restricted to a single thread and to a time limit of 30 minutes.

\paragraph{Instance generation:}

We generated synthetic instances. For each selected combination of~$\numbins$ and~$\sizebin$, 10 instances were generated;  in each individual instance, all bins have the same capacity.  Item sizes are randomly generated according to the following distribution: size 2 with probability 0.4; size 3 with probability 0.3; size 4 with probability 0.2; and size 5 with probability 0.1. This distribution was selected because of the authors' experience with group seating optimization applications.

Items are generated uniformly and independently at random from the above distribution until 85\% of the overall capacity is occupied. We then sequentially assign colors to the items by selecting $p$ items to form each color class.  With probability 0.6 we selected $p \in \{2,3,4\}$ items, and with probability 0.4 we select between $p \in \{5,6,7,8\}$, in both cases sampled uniformly at random. If only one item remains, we assign it to the last color. We restrict our experiments to scenarios where~$\binsize \geq 8$, as instances with smaller bins can be efficiently solved (see~Proposition~\ref{prop:easy2bins}).  We generate instances with $k \in \{10,20,30\}$ and $B \in \{8,10,12\}$, and additionally instances with $k=50$ and $B=12$ to evaluate how well \textbf{ANF} scales.

\paragraph{Results:}

The results of our experiments are shown in Table~\ref{tab:results} in aggregation.  Each row corresponds to a configuration of instances with $k,B$, as indicated by the first and second columns.  The next eight columns report solution statistics,
first for~\textbf{IP} and then for~\textbf{ANF}.  In sequence, we report the average solution times for those instances that were solved within 1800 seconds, with the number of instances solved within 1800 seconds in superscript, the average ending lower bound, the average ending upper bound, and the average gap. 

\begin{table}[]
	\centering
	\footnotesize
	\begin{tabular}{|c|c|c|c|c|c|c|c|c|c|}
		\hline
		\multicolumn{2}{|c|}{Instances} & \multicolumn{4}{c|}{IP} & \multicolumn{4}{c|}{ANP} \\
		\hline
		$\numbins$  & $\sizebin$ & Time & LB & UB & Gap & Time & LB & UB & Gap \\
		\hline
		10 & 8&	$59.05^{10}$&	10.1 &	 10.1 &	 0.0&	$0.10^{10}$&	10.1 &	 10.1 & 0.0\\
		10 & 10&	$525.52^{8}$&	11.2 &	 11.4 &	 1.6&	$0.50^{10}$&	11.4 &	 11.4 &	 0.0\\
		10 & 12&	$428.96^{6}$&	11.2 &	 11.7 &	 4.0&	$66.15^{10}$&	11.7 &	 11.7 &	 0.0\\
		20 & 8&	-&	16.4 &	 21.2 &	 22.3&	$0.98^{10}$&	21.2 &	 21.2 &	 0.0\\
		20 & 10&	-&	18.9 &	 22.6 &	 16.2&	$86.05^{10}$&	22.6 &	 22.6 & 0.0\\
		20 & 12&	-&	19.4 &	 23.7 &	 17.9&	$404.65^{9}$&	23.5 &	 23.7 &	 0.9\\
		30 & 8&	-&	23.1 &	 31.9 &	 27.6&	$37.52^{10}$&	31.9 &	 31.9 &	 0.0\\
		30 & 10&	-&	27.7 &	 34.1 &	 18.5&	$802.51^{9}$&	34.0 &	 34.1 &	 0.2\\
		30 & 12&	-&	28.0 &	 34.6 &	 18.9&	-&	30.4 &	 34.8 &	 12.4\\
		50 & 10&	-&	40.9 &	 56.2 &	 27.3&	-&	50.0 &	 56.2 &	 10.9\\
		\hline
	\end{tabular}
	\caption{Aggregate summary of results.}
	\label{tab:results}
\end{table}

We see a considerable superiority of~\textbf{ANF} over~\textbf{IP}, both in terms of gap and running time. \textbf{IP} solves only those instances with $k=10$ (and only solves 24 of the 30 instanes with this $k$) while \textbf{ANF} solves all instances with $k=10$ and $k=20$, and even 10 with $k=30$.  Additionally, the ending gap and quality of solutions are significantly better, even for those instances unsolved by both.

A depiction of the solution time and ending gaps is provided in the Figure~\ref{f:cdpp} through a cumulative distribution plot of performance. For both algorithms, the left half provides a plot with height equal to the cumulative number of instances solved at the time given on the horizontal axis. In the right half, the height of the plot corresponds to the number of instances with at most the optimality gap given on the horizontal axis by the time limit of 1800 seconds. Figure~\ref{f:cdpp} more readily depicts the overall performance of~\textbf{ANF}. After any amount of time, \textbf{ANF} solves more instances than \textbf{IP}, with smaller gaps at time limit.

\begin{figure}[t!]
	\centering
	\tiny
	\includegraphics[scale=0.3]{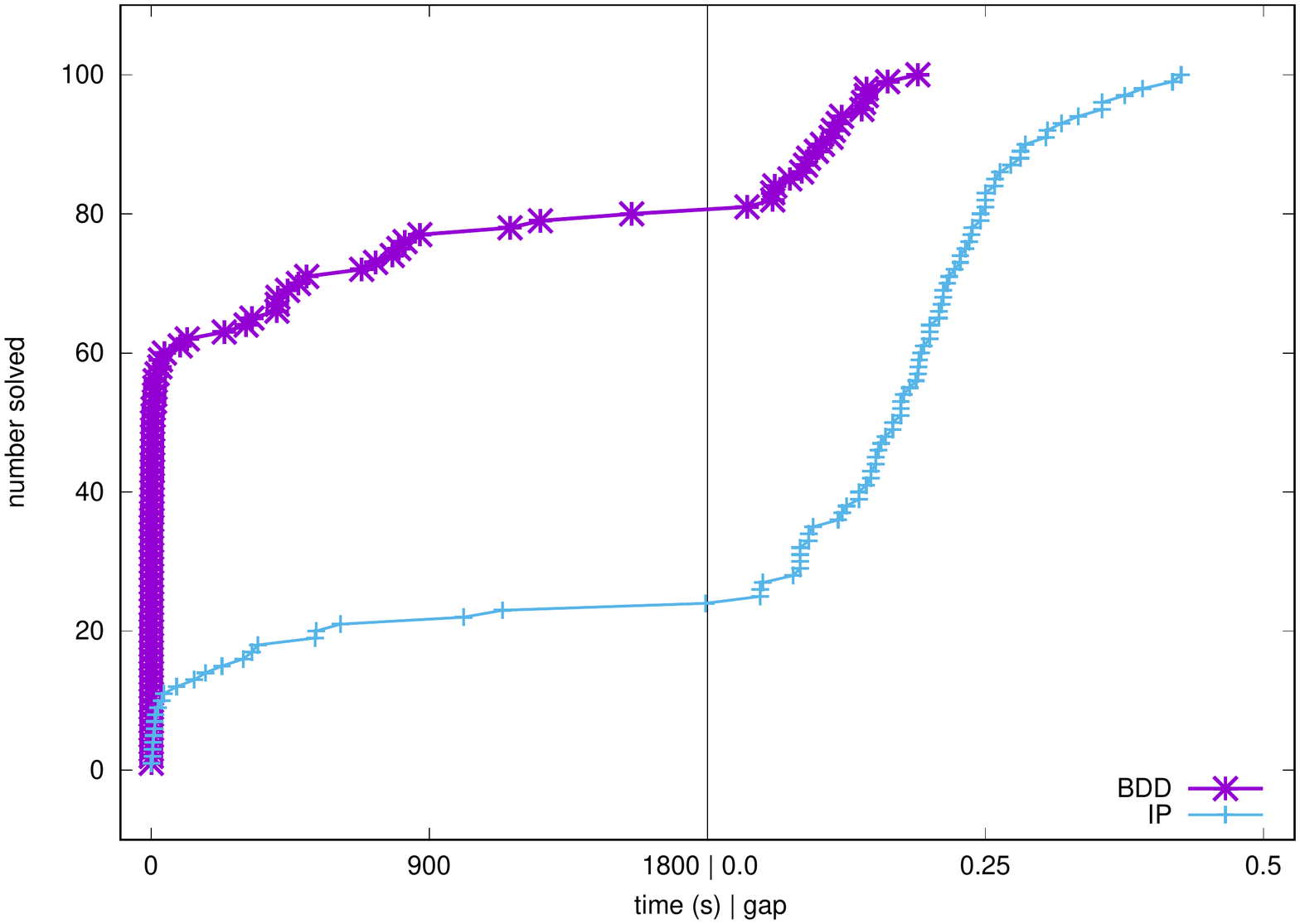}
	\captionof{figure}{Cumulative distribution plot comparing \textbf{BDD} with \textbf{IP}.}
	\label{f:cdpp}
\end{figure}

\section{Conclusion and future work}\label{sec:Conclusion}

In this work, we have introduced the \emph{bin packing with minimum color fragmentation} and presented an algorithm consisting of the integration of decision diagrams and mixed-integer linear programming. Namely, we showed how to represent the assignment of items to individual bins as binary decision diagrams and formulated the integration of the sub-problems using a network flow model. Our computational experiments have shown 
that the proposed algorithm scales well and is clearly superior to a direct formulation of~\BPMCF{}.

In future work, we intend to investigate the performance of the proposed algorithm in real-world scenarios. Additionally, we also would like to investigate alternative advanced solution approaches that have been successfully applied to other variants of the colored bin packing problem, such as branch and price.

	\bibliographystyle{plainnat}
	\bibliography{references}

\begin{thebibliography}{40}
\providecommand{\natexlab}[1]{#1}
\providecommand{\url}[1]{\texttt{#1}}
\expandafter\ifx\csname urlstyle\endcsname\relax
  \providecommand{\doi}[1]{doi: #1}\else
  \providecommand{\doi}{doi: \begingroup \urlstyle{rm}\Url}\fi

\bibitem[Andersen et~al.(2007)Andersen, Hadzic, Hooker, and
  Tiedemann]{andersen2007constraint}
Henrik~Reif Andersen, Tarik Hadzic, John~N Hooker, and Peter Tiedemann.
\newblock A constraint store based on multivalued decision diagrams.
\newblock In \emph{International Conference on Principles and Practice of
  Constraint Programming}, pages 118--132. Springer, 2007.

\bibitem[Balogh et~al.(2012)Balogh, B{\'e}k{\'e}si, Dosa, Kellerer, and
  Tuza]{balogh2012black}
J{\'a}nos Balogh, J{\'o}zsef B{\'e}k{\'e}si, Gyorgy Dosa, Hans Kellerer, and
  Zsolt Tuza.
\newblock Black and white bin packing.
\newblock In \emph{International Workshop on Approximation and Online
  Algorithms}, pages 131--144. Springer, 2012.

\bibitem[Balogh et~al.(2015)Balogh, B{\'e}k{\'e}si, D{\'o}sa, Epstein,
  Kellerer, and Tuza]{balogh2015online}
J{\'a}nos Balogh, J{\'o}zsef B{\'e}k{\'e}si, Gy{\"o}rgy D{\'o}sa, Leah Epstein,
  Hans Kellerer, and Zsolt Tuza.
\newblock Online results for black and white bin packing.
\newblock \emph{Theory of Computing Systems}, 56\penalty0 (1):\penalty0
  137--155, 2015.

\bibitem[Behle(2008)]{Behle2008}
Markus Behle.
\newblock On threshold bdds and the optimal variable ordering problem.
\newblock \emph{Journal of Combinatorial Optimization}, 16\penalty0
  (2):\penalty0 107--118, Aug 2008.
\newblock ISSN 1573-2886.
\newblock \doi{10.1007/s10878-007-9123-z}.

\bibitem[Bergman et~al.(2016{\natexlab{a}})Bergman, Cire, van Hoeve, and
  Hooker]{DDBook}
D.~Bergman, A.A. Cire, W.J. van Hoeve, and J.~Hooker.
\newblock \emph{Decision diagrams for optimization}.
\newblock Springer, 2016{\natexlab{a}}.

\bibitem[Bergman and Cire(2016{\natexlab{a}})]{10.1007/978-3-319-44953-1_6}
David Bergman and Andre~A. Cire.
\newblock Multiobjective optimization by decision diagrams.
\newblock In Michel Rueher, editor, \emph{Principles and Practice of Constraint
  Programming}, pages 86--95, Cham, 2016{\natexlab{a}}. Springer International
  Publishing.
\newblock ISBN 978-3-319-44953-1.

\bibitem[Bergman and Cire(2016{\natexlab{b}})]{bergman2016decomposition}
David Bergman and Andre~A Cire.
\newblock Decomposition based on decision diagrams.
\newblock In \emph{International Conference on AI and OR Techniques in
  Constriant Programming for Combinatorial Optimization Problems}, pages
  45--54. Springer, 2016{\natexlab{b}}.

\bibitem[Bergman and Cire(2018)]{BerCir17}
David Bergman and Andre~A. Cire.
\newblock Discrete nonlinear optimization by state-space decompositions.
\newblock \emph{Management Science}, 64\penalty0 (10):\penalty0 4700--4720,
  2018.
\newblock \doi{10.1287/mnsc.2017.2849}.

\bibitem[Bergman and Lozano(2018)]{bergman2018decision}
David Bergman and Leonardo Lozano.
\newblock Decision diagram decomposition for quadratically constrained binary
  optimization.
\newblock 2018.

\bibitem[Bergman et~al.(2011)Bergman, van Hoeve, and Hooker]{BerHoeHoo11}
David Bergman, Willem-Jan van Hoeve, and John~N. Hooker.
\newblock Manipulating mdd relaxations for combinatorial optimization.
\newblock In Tobias Achterberg and J.~Christopher Beck, editors,
  \emph{Integration of AI and OR Techniques in Constraint Programming for
  Combinatorial Optimization Problems}, pages 20--35, Berlin, Heidelberg, 2011.
  Springer Berlin Heidelberg.
\newblock ISBN 978-3-642-21311-3.

\bibitem[Bergman et~al.(2012)Bergman, Cire, van Hoeve, and
  Hooker]{10.1007/978-3-642-29828-8_3}
David Bergman, Andre~A. Cire, Willem-Jan van Hoeve, and John~N. Hooker.
\newblock Variable ordering for the application of bdds to the maximum
  independent set problem.
\newblock In Nicolas Beldiceanu, Narendra Jussien, and {\'E}ric Pinson,
  editors, \emph{Integration of AI and OR Techniques in Contraint Programming
  for Combinatorial Optimzation Problems}, pages 34--49, Berlin, Heidelberg,
  2012. Springer Berlin Heidelberg.
\newblock ISBN 978-3-642-29828-8.

\bibitem[Bergman et~al.(2016{\natexlab{b}})Bergman, Cire, van Hoeve, and
  Hooker]{Bergman:2016:DOD:3214682.3214686}
David Bergman, Andre~A. Cire, Willem-Jan van Hoeve, and J.~N. Hooker.
\newblock Discrete optimization with decision diagrams.
\newblock \emph{INFORMS J. on Computing}, 28\penalty0 (1):\penalty0 47--66,
  February 2016{\natexlab{b}}.
\newblock ISSN 1526-5528.
\newblock URL \url{https://doi.org/10.1287/ijoc.2015.0648}.

\bibitem[Bergman et~al.(2018)Bergman, Bodur, Cardonha, and Cire]{1802.08637}
David Bergman, Merve Bodur, Carlos Cardonha, and Andre~A. Cire.
\newblock Network models for multiobjective discrete optimization, 2018.

\bibitem[B{\"o}hm et~al.(2014)B{\"o}hm, Sgall, and Vesel{\`y}]{bohm2014online}
Martin B{\"o}hm, Ji{\v{r}}{\'\i} Sgall, and Pavel Vesel{\`y}.
\newblock Online colored bin packing.
\newblock In \emph{International Workshop on Approximation and Online
  Algorithms}, pages 35--46. Springer, 2014.

\bibitem[Bryant(1986)]{Bry86}
Randal~E. Bryant.
\newblock Graph-based algorithms for boolean function manipulation.
\newblock \emph{IEEE Trans. Comput.}, 35\penalty0 (8):\penalty0 677--691,
  August 1986.
\newblock ISSN 0018-9340.
\newblock \doi{10.1109/TC.1986.1676819}.

\bibitem[Bryant(1992)]{Bry92}
Randal~E. Bryant.
\newblock Symbolic boolean manipulation with ordered binary-decision diagrams.
\newblock \emph{{ACM} Comput. Surv.}, 24\penalty0 (3):\penalty0 293--318, 1992.
\newblock \doi{10.1145/136035.136043}.

\bibitem[Dawande et~al.(2001)Dawande, Kalagnanam, and
  Sethuraman]{dawande2001variable}
Milind Dawande, Jayant Kalagnanam, and Jay Sethuraman.
\newblock Variable sized bin packing with color constraints.
\newblock \emph{Electronic Notes in Discrete Mathematics}, 7:\penalty0
  154--157, 2001.

\bibitem[Edmonds(1965)]{edmonds1965paths}
Jack Edmonds.
\newblock Paths, trees, and flowers.
\newblock \emph{Canadian Journal of mathematics}, 17\penalty0 (3):\penalty0
  449--467, 1965.

\bibitem[Elhedhli et~al.(2011)Elhedhli, Li, Gzara, and
  Naoum-Sawaya]{elhedhli2011branch}
Samir Elhedhli, Lingzi Li, Mariem Gzara, and Joe Naoum-Sawaya.
\newblock A branch-and-price algorithm for the bin packing problem with
  conflicts.
\newblock \emph{INFORMS Journal on Computing}, 23\penalty0 (3):\penalty0
  404--415, 2011.

\bibitem[Garey and Johnson(1978)]{garey1978strong}
Michael~R Garey and David~S Johnson.
\newblock ``strong'' np-completeness results: Motivation, examples, and
  implications.
\newblock \emph{Journal of the ACM (JACM)}, 25\penalty0 (3):\penalty0 499--508,
  1978.

\bibitem[Garey and Johnson(2002)]{garey2002computers}
Michael~R Garey and David~S Johnson.
\newblock \emph{Computers and intractability}, volume~29.
\newblock wh freeman New York, 2002.

\bibitem[Gendreau et~al.(2004)Gendreau, Laporte, and
  Semet]{gendreau2004heuristics}
Michel Gendreau, Gilbert Laporte, and Fr{\'e}d{\'e}ric Semet.
\newblock Heuristics and lower bounds for the bin packing problem with
  conflicts.
\newblock \emph{Computers \& Operations Research}, 31\penalty0 (3):\penalty0
  347--358, 2004.

\bibitem[Gurobi~Optimization(2018)]{gurobi}
LLC Gurobi~Optimization.
\newblock Gurobi optimizer reference manual, 2018.
\newblock URL \url{http://www.gurobi.com}.

\bibitem[Hadzic et~al.(2008)Hadzic, Hooker, O'Sullivan, and
  Tiedemann]{HadHooOsulTieStu08}
Tarik Hadzic, John~N. Hooker, Barry O'Sullivan, and Peter Tiedemann.
\newblock Approximate compilation of constraints into multivalued decision
  diagrams.
\newblock In Peter~J. Stuckey, editor, \emph{Principles and Practice of
  Constraint Programming}, pages 448--462, Berlin, Heidelberg, 2008. Springer
  Berlin Heidelberg.
\newblock ISBN 978-3-540-85958-1.

\bibitem[Jansen(1999)]{jansen1999approximation}
Klaus Jansen.
\newblock An approximation scheme for bin packing with conflicts.
\newblock \emph{Journal of combinatorial optimization}, 3\penalty0
  (4):\penalty0 363--377, 1999.

\bibitem[Jansen and {\"O}hring(1997)]{jansen1997approximation}
Klaus Jansen and Sabine {\"O}hring.
\newblock Approximation algorithms for time constrained scheduling.
\newblock \emph{Information and computation}, 132\penalty0 (2):\penalty0
  85--108, 1997.

\bibitem[Karp(1972)]{karp1972reducibility}
Richard~M Karp.
\newblock Reducibility among combinatorial problems.
\newblock In \emph{Complexity of computer computations}, pages 85--103.
  Springer, 1972.

\bibitem[Kochetov and Kondakov(2017)]{kochetov2017vns}
Y~Kochetov and A~Kondakov.
\newblock Vns matheuristic for a bin packing problem with a color constraint.
\newblock \emph{Electronic Notes in Discrete Mathematics}, 58:\penalty0 39--46,
  2017.

\bibitem[Kondakov and Kochetov(2018)]{kondakov2018core}
Artem Kondakov and Yury Kochetov.
\newblock A core heuristic and the branch-and-price method for a bin packing
  problem with a color constraint.
\newblock In \emph{International Conference on Optimization Problems and Their
  Applications}, pages 309--320. Springer, 2018.

\bibitem[Lozano et~al.(2018)Lozano, Bergman, and Smith]{LozBerSmi18}
Leonardo Lozano, David Bergman, and J.~Cole Smith.
\newblock On the consistent path problem.
\newblock 2018.

\bibitem[Matsumoto et~al.(2018)Matsumoto, Hatano, and
  Takimoto]{matsumoto2018decision}
Kosuke Matsumoto, Kohei Hatano, and Eiji Takimoto.
\newblock Decision diagrams for solving a job scheduling problem under
  precedence constraints.
\newblock In \emph{LIPIcs-Leibniz International Proceedings in Informatics},
  volume 103. Schloss Dagstuhl-Leibniz-Zentrum fuer Informatik, 2018.

\bibitem[Miller and Drechsler(1998)]{miller1998implementing}
D~Michael Miller and Rolf Drechsler.
\newblock Implementing a multiple-valued decision diagram package.
\newblock In \emph{Multiple-Valued Logic, 1998. Proceedings. 1998 28th IEEE
  International Symposium on}, pages 52--57. IEEE, 1998.

\bibitem[Muritiba et~al.(2010)Muritiba, Iori, Malaguti, and
  Toth]{muritiba2010algorithms}
Albert E~Fernandes Muritiba, Manuel Iori, Enrico Malaguti, and Paolo Toth.
\newblock Algorithms for the bin packing problem with conflicts.
\newblock \emph{INFORMS Journal on computing}, 22\penalty0 (3):\penalty0
  401--415, 2010.

\bibitem[Peeters and Degraeve(2004)]{peeters2004co}
Marc Peeters and Zeger Degraeve.
\newblock The co-printing problem: A packing problem with a color constraint.
\newblock \emph{Operations Research}, 52\penalty0 (4):\penalty0 623--638, 2004.

\bibitem[Raghunathan et~al.(2018)Raghunathan, Bergman, Hooker, Serra, and
  Kobori]{1807.09676}
Arvind~U Raghunathan, David Bergman, John Hooker, Thiago Serra, and Shingo
  Kobori.
\newblock Seamless multimodal transportation scheduling, 2018.

\bibitem[Sadykov and Vanderbeck(2013)]{sadykov2013bin}
Ruslan Sadykov and Fran{\c{c}}ois Vanderbeck.
\newblock Bin packing with conflicts: a generic branch-and-price algorithm.
\newblock \emph{INFORMS Journal on Computing}, 25\penalty0 (2):\penalty0
  244--255, 2013.

\bibitem[Shachnai and Tamir(2001)]{shachnai2001polynomial}
Hadas Shachnai and Tami Tamir.
\newblock Polynomial time approximation schemes for class-constrained packing
  problems.
\newblock \emph{Journal of Scheduling}, 4\penalty0 (6):\penalty0 313--338,
  2001.

\bibitem[Shachnai and Tamir(2004)]{shachnai2004tight}
Hadas Shachnai and Tami Tamir.
\newblock Tight bounds for online class-constrained packing.
\newblock \emph{Theoretical Computer Science}, 321\penalty0 (1):\penalty0
  103--123, 2004.

\bibitem[Trick(2003)]{Trick2003}
Michael~A. Trick.
\newblock A dynamic programming approach for consistency and propagation for
  knapsack constraints.
\newblock \emph{Annals of Operations Research}, 118\penalty0 (1):\penalty0
  73--84, Feb 2003.
\newblock ISSN 1572-9338.
\newblock \doi{10.1023/A:1021801522545}.

\bibitem[Xavier and Miyazawa(2008)]{xavier2008class}
Eduardo~C Xavier and Fl{\'a}vio~Keidi Miyazawa.
\newblock The class constrained bin packing problem with applications to
  video-on-demand.
\newblock \emph{Theoretical Computer Science}, 393\penalty0 (1-3):\penalty0
  240--259, 2008.

\end{thebibliography}

\end{document}